\documentclass[conference]{IEEEtran}
\IEEEoverridecommandlockouts

\usepackage{cite}
\usepackage{graphicx, xcolor} 
\usepackage{amsfonts, amsmath, amssymb, amsthm, dsfont}
\usepackage[nocomma]{optidef}
\usepackage{enumitem}

\newtheorem{corollary}{\textbf{Corollary}}

\newtheorem{lemma}{\textbf{Lemma}}

\newtheorem{remark}{\textbf{Remark}}

\DeclareMathOperator*{\argmin}{arg\,min}

\title{Emergent Cooperation for Energy-efficient Connectivity via Wireless Power Transfer\vspace{-10pt}}

\author{Winston Hurst, Anurag Pallaprolu and Yasamin Mostofi 
\thanks{Winston Hurst, Anurag Pallaprolu, and Yasamin Mostofi are with the Department of Electrical and Computer Engineering, University of California, Santa Barbara, USA (email: winstonhurst@ece.ucsb.edu, apallaprolu@ucsb.edu, ymostofi@ece.ucsb.edu). This work was supported in part by NSF RI award 2008449 and in part by ONR award N00014-23-1-2715.}
}

\begin{document}
\IEEEpubid{\begin{minipage}{\textwidth}\centering
    {\vspace{1in}\footnotesize \copyright 2024 IEEE.  Personal use of this material is permitted.  Permission from IEEE must be obtained for all other uses, in any current or future media, including reprinting/republishing this material for advertising or promotional purposes, creating new collective works, for resale or redistribution to servers or lists, or reuse of any copyrighted component of this work in other works.}
\end{minipage}}

\maketitle

\begin{abstract}
This paper addresses the challenge of incentivizing energy-constrained, non-cooperative user equipment (UE) to serve as cooperative relays. We consider a source UE with a non-line-of-sight channel to an access point (AP), where direct communication may be infeasible or may necessitate a substantial transmit power. Other UEs in the vicinity are viewed as relay candidates, and our aim is to enable energy-efficient connectivity for the source, while accounting for the self-interested behavior and private channel state information of these candidates, by allowing the source to ``pay" the candidates via wireless power transfer (WPT). We propose a cooperation-inducing protocol, inspired by Myerson auction theory, which ensures that candidates truthfully report power requirements while minimizing the expected power used by the source. Through rigorous analysis, we establish the regularity of valuations for lognormal fading channels, which allows for the efficient determination of the optimal source transmit power. Extensive simulation experiments, employing real-world communication and WPT parameters, validate our theoretical framework. Our results
demonstrate over 71\% reduction in outage probability with as few as 4 relay candidates, compared to the non-cooperative scenario, and as much as 70\% source power savings compared to a baseline approach, highlighting the efficacy of our proposed methodology.
\end{abstract}
\vspace{-0.2in}
\section{Introduction}
The forthcoming 6G wireless standard promises to revolutionize wireless communication by delivering significant improvements in reliability, data rates, and energy efficiency~\cite{IMT-2030}. Achieving the desired technical performance depends on several emerging technologies and paradigms, which in turn bring their own challenges. For instance, high frequency communication (e.g., mmWave, THz) enables vastly increased capacity, but suffers from acute shadowing and penetration losses \cite{Hemadeh2018CST}. At the same time, given the ever-increasing number of connected devices, there is a growing need for scalable system design, with reduced reliance on centralized coordination mechanisms, to produce the desired performance.

This paper proposes the use of wireless power transfer (WPT) \cite{Hossain2019Access} to induce cooperative behavior among non-cooperative communication devices, with the end goal of enabling connectivity and/or reducing energy consumption. More specifically, we consider a setting in which a large-scale blockage (\textit{e.g.,} a wall) obstructs the line-of-sight (LoS) link between a user equipment (UE) and an access point (AP), so that direct communication is challenging. To reduce the power required to connect to the AP, the blocked user may offer to send a payment of power (via WPT) to a more advantageously placed UE, in exchange for relay services. The core problem we address in this paper is the identification of the optimal relay candidate from a potentially extensive set of options, along with the determination of the associated power cost.

Extensive prior work considers the use of wireless power transfer and related technologies to enable relaying in cooperative settings, in contexts such as IoT networks or device-to-device communication \cite{Li2021IoTJ, salim2022IoTaIS}. Other works, such as \cite{Xu2022ICC}, more explicitly acknowledge the need to incentivize the relay operation, but still assume a basic level of cooperation.
\begin{figure}
    \centering
    \includegraphics[width =0.85\linewidth]{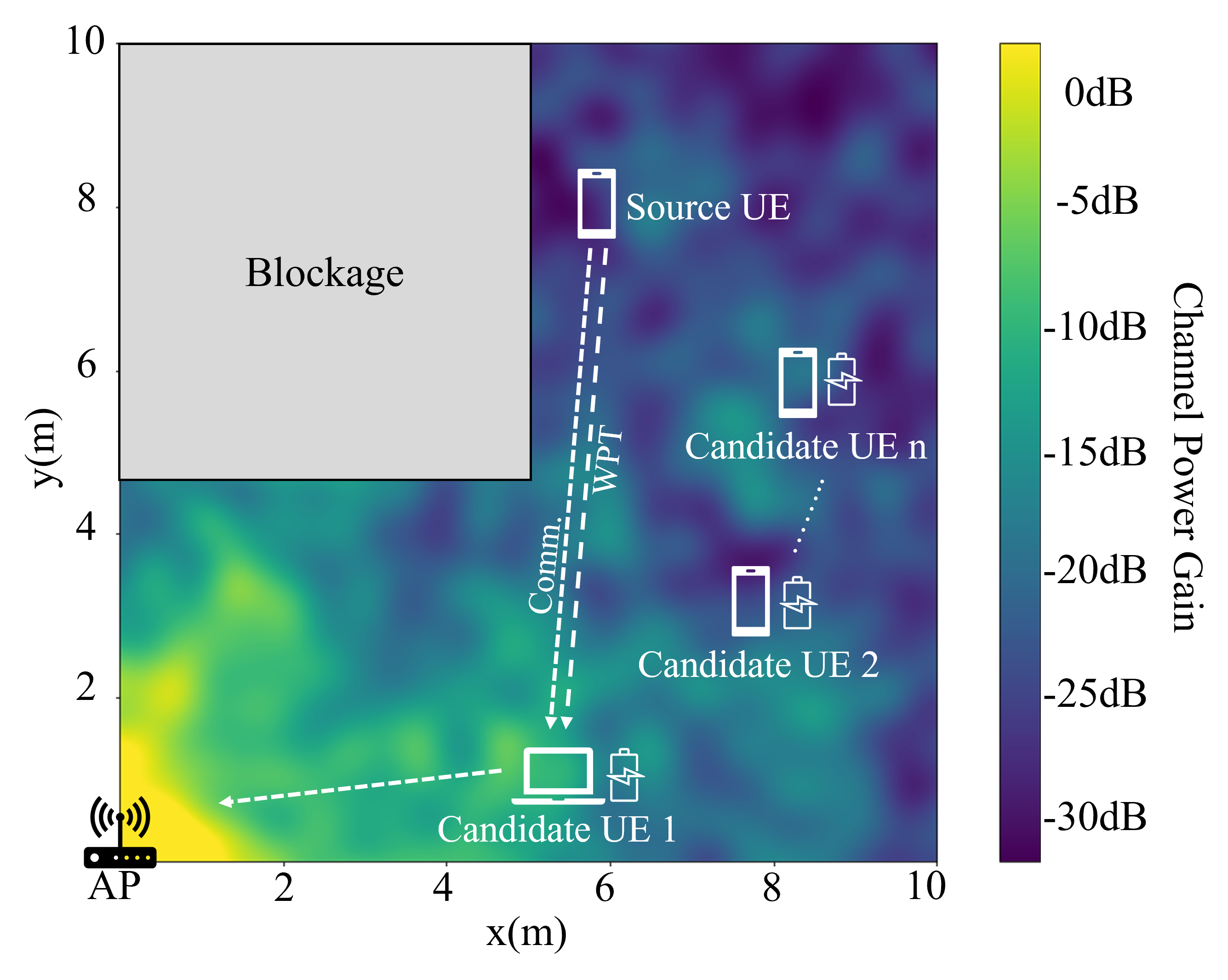}
    \vspace{-8pt}
    \caption{Fundamental scenario considered in this work: A single \textit{source} UE attempts to communicate with a blocked AP using one of the $n$ \textit{relay candidate} UEs, incentivizing cooperation with wireless power transfer.}
    \vspace{-18pt}
    \label{fig:system_setup}
\end{figure}
We instead consider fully \textit{non-cooperative} relay candidate devices that are energy limited, and possibly mobile (\textit{e.g.,} a cellular phone or autonomous vehicle), thus requiring an incentive to provide the relay service. We approach the scenario from a game-theoretic perspective, inspired by several-decades-old seminal work on auction theory, to develop a protocol which creates emergent cooperation among non-cooperative entities.

Game theory has long been an important tool in the design and analysis of communication systems \cite{GTforWCN2012HanNiyatoSaadBasarHJorungnes, TRO19_MuralidharanMostofi}, and recent years have seen the use of auctions linked to WPT in conjunction with energy markets \cite{Ni2019ICCT}. These works, however, apply auctions in the fiscal domain, using money as the payment medium. In contrast, we propose the use of wireless power itself as an RF-native ``currency", thus eliminating the need for complex configuration and integration between the physical wireless system and financial infrastructure.

Our contributions are succinctly stated below:

\begin{itemize}[leftmargin=0.15in]
    \item We consider a source UE with a non-line-of-sight (NLoS) channel to an AP, so that communication either requires significant transmit power or is impossible altogether. This source seeks to use WPT to incentivize non-cooperative, energy limited UEs to act as cooperative relays (see Fig.~\ref{fig:system_setup}).
    \item We pose the problem of finding a cooperation-inducing protocol which minimizes the energy consumption of the source, while accounting for the self-interested rationality of the relay candidates and their private channel state information, which is unavailable to the source.
    \item We show that the solution to our problem is given by a Myerson auction \cite{Myerson1981MOR}, with valuations determined by the minimum source transmit power needed for the relay candidates to be neutral towards cooperation. We provide a rigorous proof that these valuations are \textit{regular} for wireless channels described by lognormal fading, which allows the optimal source transmit power to be found with a simple bisection search.
    \item We validate our theoretical findings through extensive simulation experiments, with communication and WPT parameter values taken from real-world experiments reported in the literature. Our results demonstrate a 71\% reduction in outage probability with as few as 4 relay candidates compared to the non-cooperative alternative, and as much as 48\% source power savings compared to a baseline approach. 
\end{itemize}

\section{System Model and Motivation}\label{sec:model}

We consider a system consisting of a single \textit{source} UE, $n$ battery-powered \textit{relay candidate} UEs, and a single AP, as depicted in Fig.~\ref{fig:system_setup}. We assume that the channel between the source and the AP is quite poor, e.g., due to a LoS blockage, so that direct communication with the AP is either impossible or requires an unreasonably large transmit power. Consequently, the source looks to induce one of the candidate UEs to act as its relay. Without loss of generality, we set the AP as the origin, and let $q_s$ and $q_{i}, \forall i \in N=\{1,..., n\}$ denote the 2D positions of the source and the $i^\text{th}$ candidate, respectively. We next describe our modeling of communication and WPT.

\subsection{Communication System Model}
To satisfy a given quality-of-service (QoS) requirement, e.g., maximum bit error rate at a minimum spectral efficiency, all communication links must maintain a minimum signal-to-noise ratio (SNR) of $\gamma_{\text{th}}$. 

The SNR across any link is given by $\gamma = PH/\sigma^2$, where $P$ is the transmit power, $H$ is the channel coefficient (given by the square of the modulus of the complex channel gain between the transmitter and receiver), and $\sigma^2$ is the noise power. We denote the channel coefficient between the source (candidate) UE and the AP as $H_{s}$ ($H_i$), while $H_{si}$ denotes the channel between the source and candidate $i$. 


Based on \cite{Hemadeh2018CST, Samimi2016VTC}, we use a reference path loss channel model with lognormal small-scale fading, and we account for large scale fading by considering different sets of channel parameters for LoS and NLoS cases. In the dB scale, this model is parameterized by a path loss intercept, $H_{\text{PL}, 0, \text{dB}}^c$, a path loss exponent, $n_{\text{PL}}^c$, and fading power, $(\sigma_{\text{SS, dB}}^c)^2$, where $c\in\{\text{L},\,\text{NL}\}$ indicates the corresponding channel parameters for the LoS and NLoS cases. We assume that $H_i$ and $H_{si}$ are LoS channels, as depicted in Fig.~\ref{fig:system_setup}, as otherwise, the candidate would almost certainly not be useful as a relay. We thus express these channels as
\vspace{-0.025in}
\begin{equation*}
    H_{l,\text{dB}} = H_{\text{PL}, 0, \text{dB}}^{\text{L}} -  10n_{\text{PL}}^\text{L}\text{log}_{10}(d_l) - H_{\sigma_{\text{SS,dB}}^\text{L},l},
    \vspace{-0.025in}
\end{equation*}
for $l\in \cup_{i=1}^n\{si,\,i\}$, with $d_l$ the corresponding distance between the two endpoints of the channel and $H_{\sigma_{\text{SS,dB}}^{\text{L}},l}\sim \mathcal{N}(0,\,(\sigma_{\text{SS,dB}}^{\text{L}})^2)$. We also assume that $H_s$ is NLoS:
\vspace{-0.025in}
\begin{equation*}
    H_{s,\text{dB}} = H_{\text{PL}, 0, \text{dB}}^{\text{NL}} -  10n_{\text{PL}}^\text{NL}\text{log}_{10}(d_s) - H_{\sigma_{\text{SS,dB}}^{\text{NL}},s},
    \vspace{-0.025in}
\end{equation*}
where $d_s = ||q_s||_2^2$ is the distance between the source and the AP, and $H_{\sigma_{\text{SS,dB}}^{\text{NL}},s}\sim \mathcal{N}(0,\,(\sigma_{\text{SS,dB}}^{\text{NL}})^2)$. The small-scale fading is taken to be independent across all channels \cite{Samimi2016VTC}.

If the source attempts to directly communicate with the AP, it must use a power of $P_{s} = \gamma_{\text{th}} \sigma^2/H_s$. However, due to the NLoS channel, this power may be quite large and may exceed the maximum transmit power of the device, $P_\text{max}$, causing an outage. Alternatively, if the transmission is cooperatively relayed via candidate $i$, the source first transmits by expending $P_{si} = \gamma_{\text{th}} \sigma^2/H_{si}$, after which the candidate transmits with power $P_{i} = \gamma_{\text{th}} \sigma^2/H_{i}$ over the LoS link with the AP.

Under certain channel conditions (specifically, if $H_{si}^{-1} + H_i^{-1} < H_s^{-1}$) this relaying setup provides improved energy-efficient connectivity. However, the gains from cooperation are not equally distributed across the source and the chosen relay candidate, as the candidate expends energy from its limited store to help meet the QoS requirement of the source. Consequently, in scenarios where the source and relay are independent, non-cooperative actors, the candidate has no incentive to participate in this scheme. 

\vspace{-0.05in}
\subsection{Wireless Power Transfer}
To achieve the utility transfer necessary to enable candidate participation in the aforementioned scenario, we propose leveraging wireless power transfer from the source to the relay candidates. More specifically, the source uses simultaneous wireless information and power transfer (SWIPT), so that the total source transmit power is given by $P^{\text{tot}} = P_{si} + P_\text{WPT}$, where $P_{\text{WPT}}$ is the power allocated for WPT. The candidate then uses a dynamic power-splitting setup~\cite{Zhou2012GLOBECOM}, with the power splitting ratio given by $P_{\text{WPT}}/P^{\text{tot}}$. The power harvested by a candidate relay $i$ is then given by $P_{i}^\text{harv} = \alpha P_\text{WPT} A_rH_{si}$ \cite{Liu2013TC}, with $\alpha\in[0,1]$ denoting the efficiency of the power harvesting circuitry, and $A_r$ denoting the effective aperture area of the receiving antenna (assumed to be identical for all candidates).

The candidate will only participate in the relay scheme if it harvests at least $P_i$ power, which covers the subsequent transmission to the AP. Thus, the minimum source transmit power that induces cooperation with candidate $i$, is
\begin{equation}\label{eq:p_si_min}
    P_{si}^{\text{min}} = P_{si} + \frac{P_i}{\alpha A_r H_{si}}.
\end{equation}
If $P^{\text{tot}} \geq P_{si}^{\text{min}}$, the candidate harvests $P_i^{\text{harv}} = \alpha (P^{\text{tot}} - P_{si}) A_r H_{si} \geq P_i$, and any surplus power can be used for charging its battery or powering on-board sensors. However, $P_{si}^{\text{min}}$ is not exactly known to the source, as we next detail in our information availability assumptions.

\subsection{Information Assumptions}
We assume that both the source and the relay candidates individually know the value of the channel coefficient between themselves and the AP, but do not have precise information about the channel between other UEs and the AP. Additionally, the channel coefficient between the source and a candidate is known to both parties, but unknown to other candidates. 

We also assume that the source and each candidate can estimate the channel, $H_{i}$, between any candidate and the AP using estimation frameworks such as \cite{TWC11}, in conjunction with information acquired from prior data collected during operation, including the coordination process presented in Section~\ref{sec:prob_form}. Specifically, the channel parameters $H_{\text{PL}, 0, \text{dB}}^{\text{L}}$, $n_{\text{PL}}^{\text{L}}$ and $\sigma_{\text{SS,dB}}^\text{L}$, along with candidate locations, $q_i$, and source location $q_s$, are known to all, but the specific realization of the fading component $H_{\sigma_{\text{SS,dB}}^\text{L},i}$ is known only to candidate $i$. Analyzing the impact of errors and differences in parameter estimation is an important avenue for future exploration.


\section{Problem Formulation}
\label{sec:prob_form}
We now mathematically state our WPT-incentivized relay scenario and construct an optimization problem for efficient selection of a relay candidate. Our objective is to design a protocol, known beforehand to all UEs, that incentivizes cooperation between the source and candidate UEs, while adhering to informational limitations. Fundamentally, the protocol should ensure that the source is able to communicate with the AP whenever possible. Additionally, as the source initiates the process, our protocol prioritizes minimizing its overall energy consumption.

To achieve this, we propose the following streamlined coordination procedure: First, the source broadcasts a signal\footnote{This signal can also be used to estimate channels, $H_{si}$, and locations $q_i$, by employing integrated sensing and communication (ISAC) or other techniques.} that queries for feasible relay candidates for data transmission. Each candidate then responds by indicating the amount of power it requires to serve as a relay for the source, which, in general, may differ from the true value, $P_{si}^\text{min}$. The source then decides on which candidate to use as a relay, if any, and the total transmit power to use, $P^{\text{tot}}$. We next model this coordination procedure as an auction game and formally state our optimization objective to make the WPT-driven relay scenario power-efficient for the source.

\subsection{Bayesian Reverse Auction Formulation}
A reverse auction \cite{Nisan_Roughgarden_Tardos_Vazirani_2007} is a procurement process in which bidders compete to win a contract specified by a buyer. Bidders submit bids, $b = \{b_i\}_{i=1}^n$, to the buyer, who then decides which bidder, determined by a function $\rho(b)$, wins the auction, and the amount to pay, denoted by $x(b)$. The pair $(\rho,\,x)$ is referred to as the \textit{auction mechanism}~\cite{Myerson1981MOR}, and is itself a function of the buyer's \textit{valuation}, $v_0$, i.e., the cost for the buyer to directly fulfill the contract (we set $\rho(b) = 0$ in such a scenario). Additionally, each bidder has a valuation, $v_i$, which is conceptually its ``breakeven" price point; if $x(b)<v_i$, the bidder's ``cost" to fulfill the contract is less than the payment. While the buyer's valuation is known to all, the bidders' valuations are private, \textit{i.e.,} unknown to the buyer or each other, and thus, the bids may differ from valuations. In a \textit{Bayesian} auction, however, the valuations are drawn from known distributions, characterized by PDFs $f_i(v_i)$ with support over $V_i$.

Our problem is well modeled as a Bayesian reverse auction, with the candidates as bidders and the source as the buyer. The private valuation of each candidate is given by $v_i = P_{si}^{\text{min}}$, the minimum source transmit power which could induce cooperation, while the valuation of the source is $v_0=\min{}\{P_{\text{max}},\, P_s\}$, which is communicated to all candidates in the coordination protocol\footnote{This is necessary to avoid market manipulation by the source. Further study of trust mechanisms in our communication system is an important direction for future work.}. If $P_s\leq P_{\text{max}}$, then $v_0$ represents the power required for the source to communicate directly with the AP. Otherwise, $v_0$ is a placeholder bid indicating the maximum amount the source can pay. 
Furthermore, $x(b) = P^{\text{tot}}\in [0, P_{\max}]$ determines the price the source pays expressed as the total transmit power for both communication and WPT (if $\rho(b)\neq 0$).

The source's knowledge of the channel statistics, $H_i,\,\forall i \in N$, provides a distribution over valuations $P_{si}^{\text{min}}$. More specifically, $P_{si}^{\text{min}}$ is a function of the realization of the channel $H_i$, known to candidate $i$ at the time of placing the bid. Decomposing $H_i$ into the deterministic path loss component, $H_{\text{PL},i}$ and a stochastic small-scale fading component, $H_{\text{SS},i}\sim \text{Lognormal}(0, \sigma_{\text{SS}}^{\text{L}})$, we have
\begin{equation*}
   P_{si}^{\text{min}} = P_{si} + \frac{\gamma_{\text{th}}\sigma^2}{H_{\text{PL}, i}H_{\text{SS},i}\alpha A_r H_{si}}.
\end{equation*}
Therefore, the PDF of $P_{si}^{\text{min}}$, denoted by $f_{i}^{\text{min}}(\cdot)$, is given as
\begin{equation}\label{eq:pdf_valuation}
\begin{split}
    f_{i}^{\text{min}}(P_{si}^{\text{min}}) &= f_{\text{ln}, \sigma_{SS}^{\text{L}}}\big(H_{\text{SS},i}(P_{si}^{\text{min}})\big) \left|\frac{d H_{\text{SS},i}(P_{si}^{\text{min}})}{dP_{si}^{\text{min}}}\right|\\    
    &=f_{\text{ln}, \sigma_{SS}^{\text{L}}}\big(H_{\text{SS},i}(P_{si}^{\text{min}})\big) \frac{H_{\text{SS},i}(P_{si}^{\text{min}})}{(P_{si}^{\text{min}} - P_{si})},\\
    \text{with } H_{\text{SS},i}&(P_{si}^{\text{min}}) = \frac{\gamma_{\text{th}}\sigma^2}{(P_{si}^{\text{min}} - P_{si})H_{\text{PL}, i}\alpha A_r H_{si}},
\end{split}
\end{equation}
where $f_{\text{ln}, \sigma_{SS}^{\text{L}}}(\cdot)$ is the PDF of $H_{\text{SS},i}$. We denote the support of $P_{si}^{\text{min}}$ as $\mathcal{P}_i$, and the support over the joint valuation space is given by $\mathcal{P}=\prod_{i=0}^n \mathcal{P}_i$.

We do not disallow bids that differ from the valuations, but the \textit{revelation principle} \cite{Nisan_Roughgarden_Tardos_Vazirani_2007} states that for any auction with untruthful bids, there is an analogous auction with truthful bids that produces the same outcome. Therefore, without loss of generality, we may focus on the case where the candidates bid their true valuations, as long as we also ensure that honest bidding of valuations is optimal for all candidates, as we do in the sequel. Thus, a candidate's bid is given by its valuation, i.e., $b_i = v_i = P_{si}^{\text{min}}$.

\subsection{Formal Problem Statement}
\label{sec:formal_prob_statement}
The problem of minimizing the transmit power of the source through a choice of the mechanism $(\rho, x)$ can be formally expressed as follows:
\vspace{-0.05in}
\begin{mini!}|s|[2]                   
    {\rho, x}                               
    { \mathbb{E}_{v\sim \mathcal{P}} [x(v)]}   
    {\label{opt}}             
    {}                                
    \addConstraint{ \mathbb{E}_{v_{-i}\sim\mathcal{P}_{-i}}\left[\mathds{1}_{i=\rho(v)}( x(v) - v_i)\right] \geq 0,\label{opt:const:individually-rational}}    
    \addConstraint{\hspace{1in} \forall i\in N, v_i\in \mathcal{P}_i,\nonumber}
    \addConstraint{\mathbb{E}_{v_{-i}\sim\mathcal{P}_{-i}}\left[\mathds{1}_{i=\rho(b)}( x(v) - v_i)\right] \geq \label{opt:const:incentive-compatible}}
    \addConstraint{\hspace{0.5in}\mathbb{E}_{v_{-i}\sim\mathcal{P}_{-i}}\left[\mathds{1}_{i=\rho(v_{-i}, s_i)}( x(v_{-i}, s_i) - v_i)\right], \nonumber}
    \addConstraint{\hspace{1in} \forall i\in N,\, \forall v_i\in \mathcal{P}_i,\,  \forall s_i\in \mathcal{P}_i. \nonumber}
\end{mini!}
where $v=\{v_i\}_{i=0}^n$, $v_{-i}$ denotes the set of all valuations except $v_i$, $(v_{-i}, s_i)$ indicates a joint valuation profile $v$ with $v_i$ replaced by $s_i$, $\mathds{1}_{\mathcal{C}}$ is a binary indicator of whether condition $\mathcal{C}$ is met, and $\mathbb{E}_{v\sim \mathcal{P}} [x(v)]$ is referred to as the \textit{expected revenue}. Constraint~(\ref{opt:const:individually-rational}) enforces that no candidate expects to expend more energy by participating in the auction than they would otherwise (\textit{individual rationality}), and Constraint~(\ref{opt:const:incentive-compatible}) ensures that bidding the true valuation is an optimal strategy for all candidates (\textit{incentive compatibility}), as required per our application of the revelation principle. In the auction context, this problem corresponds to designing a revenue-maximizing auction, as described in the next section.


\textbf{Perfect Information Lower Bound:} To establish a benchmark for evaluation, we introduce a baseline lower bound for comparison. The \textit{perfect information} solution minimizes the source's power while respecting incentive compatibility, but ignores informational constraints, i.e., the source has perfect knowledge of $H_i$ for all $i$. This is denoted as $P^{\text{lb}} = \min_{i\in N }P_{si}^{\text{min}}$, and corresponds to the source making a take-it-or-leave-it offer of $P^{\text{lb}}-\epsilon$ for an arbitrarily small $\epsilon$, which will be accepted by $i^{\text{lb}} = \text{arg}\min_{i\in N} P_{si}^{\text{min}}$.

\section{Solution via Myerson Auctions}\label{sec:solution}
To solve the optimization problem of (\ref{opt}), we employ the seminal results presented in a several-decade-old work on auction theory to design the mechanism that maximizes the expected revenue of the source, while respecting individual-rationality and incentive-compatibility constraints. In \cite{Myerson1981MOR}, Myerson demonstrated the construction of such an optimal mechanism for forward auctions when the valuation distributions of the bidders satisfy certain technical constraints. We next show how to extend the optimal Myerson mechanism to reverse auctions, and then prove that our scenario satisfies the corresponding technical constraints.

\subsection{Buyer-Optimal Reverse Auctions}
To describe the optimal auction, we first introduce the \textit{virtual valuation}, $c_i(v_i)$, which for a reverse auction we define as follows: 
\begin{equation*}
    c_i(v_i) = v_i + \frac{F_i(v_i)}{f_i(v_i)}
\end{equation*}
where $F_i(v_i)$ and $f_i(v_i)$ are the CDF and PDF of $v_i$, respectively. We call $f$ a \textit{regular distribution} if $d c_i(v_i)/d v_i > 0,\,\forall v_i \in V_i$, \textit{i.e.,} the virtual valuation is strictly increasing in the valuation over the support. We now state the key result:
\begin{lemma}\label{lma:myerson}
    If all bidder valuations are independently distributed according to regular distributions, $f_i$, then for an auction with valuations $v=\{v_i\}_{i=0}^n$, the revenue-maximizing mechanism for a reverse auction is given by 
    \begin{equation}\label{eq:opt_rho}
        \rho^*(v) = \begin{cases}
            0 & \text{If } v_0\ <\ c_i(v_i)\ \forall i\in N\\
            \argmin_{i\in N} c_i(v_i) & \text{Otherwise}
        \end{cases}
    \end{equation}
    and 
    \begin{equation}\label{eq:opt_x}
        x^*(v) = \begin{cases}
            0 & \text{If } \rho^*(v) = 0\\
            z(\rho^*(v)) & \text{Otherwise}
        \end{cases} 
    \end{equation}
    with $z(i) = \sup \{s\,|\,c_i(s)<v_0\, \text{ and }\, c_{i}(s)< c_j(v_j),\ \forall j\neq i\}$
\end{lemma}
\begin{proof}
    We note that the reverse auction may be transformed into a forward auction with valuations $\{-v_i\}_{i=0}^n$ and payments $-x(v)$. The result then follows from \cite[Lemma 3]{Myerson1981MOR}.
\end{proof}
Therefore, the optimal Myerson reverse auction assigns the object to the bidder, $\rho(v) = i^*$, with the lowest virtual valuation, $c_{i^*}(v_{i^*})$, or the buyer, if its valuation, $v_0$, is lower than all virtual valuations. The payment $x^*(v)$ is the highest valuation that would still result in $i^*$ winning the auction, which motivates the placement of honest bids from the participants. We note that while, in general, a closed form expression for $z$ cannot be found, a simple bisection search can quickly find a solution, given the monotonicity of $c_i$.

\subsection{Optimal WPT Auction}
\begin{figure*}
    \centering
    \includegraphics[width =\linewidth]{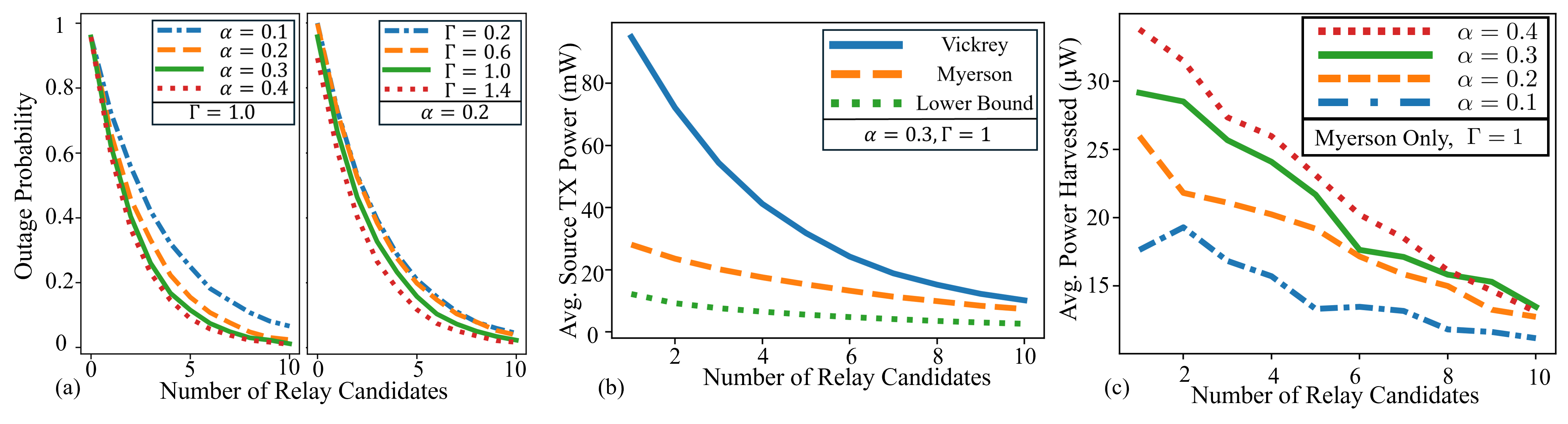}
    \vspace{-20pt}
    \caption{Sample simulation results: (a) Consistent reduction in outage probability with a Myerson auction-based relay selection, incentivized by WPT for (left) various values of WPT circuitry efficiency, $\alpha$ with $\Gamma = 1$ and (right) several scales of small-scale fading variance, $\Gamma$ with $\alpha=0.3$. (b) Higher relative source power savings of the Myerson mechanism with increasing number of relay candidates (c) Variation in harvested power with WPT efficiency ($\alpha$) as the number of relay candidates grows, with an overall reduction due to higher source power savings. See color PDF for optimal viewing.}
    \vspace{-10pt}
    \label{fig:allResults}
    \vspace{-0.1in}
\end{figure*}

To apply Lemma~\ref{lma:myerson}, we must prove that the valuations of the candidates, $v_i =P_{si}^{\text{min}}$, are regular. Using the dependence on the log-normally distributed fading component of the channel developed in Eq.~(\ref{eq:pdf_valuation}), we show in the following lemma that this is indeed the case.
 \begin{lemma}\label{lma:regularity}
     The distribution for valuations $v_i = P_{si}^{\text{min}}$, as given by Eq.~(\ref{eq:pdf_valuation}), is regular.
 \end{lemma}
\begin{proof}
    From Eq.~(\ref{eq:pdf_valuation}), it can be shown that
    \begin{equation*}
    \begin{split}
        \frac{d c_i(v_i)}{d v_i} &=2 - \frac{F_{i}(v_i) }{f_{i}(v_i)^2} \frac{d f_{i}(v_i)}{d v_i}\\
        &= 2+\frac{1-F_{\text{ln}, \sigma_{SS}^{\text{L}}}(H_{\text{SS},i}(v_i)) }{f_{\text{ln}, \sigma_{SS}^{\text{L}}}(H_{\text{SS},i}(v_i)) H_{\text{SS},i}(v_i)} \\
        &\quad \quad \quad \quad \times[1-\frac{\text{ln}(H_{\text{SS},i}(v_i))}{(\sigma_{\text{SS}}^{\text{L}})^2}]\\
        &=2+\frac{1-\Phi(z)}{\phi(z)}[\sigma_{SS}^{\text{L}} - z]
    \end{split}
    \end{equation*}
    where $\Phi$ and $\phi$ are, respectively, the CDF and PDF of the standard normal distribution, and $z = \text{ln}(H_{\text{SS},i}(v_i))/(\sigma_{\text{SS}}^{\text{L}})$. By inspection, the above expression is positive for $z\leq0$. For $z>0$, we note that the ratio $r(z) = \big(1-\Phi(z)\big)/\phi(z)$ is the Mills ratio for the standard normal distribution, and it can be shown that 
    $r(z) < 1/z,\;\forall z>0$,\, \cite{Gordon1941AMS}. Therefore,
    \begin{equation*}
            2 + \frac{1-\Phi(z)}{\phi(z)}[\sigma_{\text{SS}}^{\text{L}}-z] > 2 - \frac{1-\Phi(z)}{\phi(z)}z>1,\, \forall z>0.\qedhere
            \vspace{-0.1in}
    \end{equation*}
\end{proof}
Having shown the regularity of the bidders' valuations, we can now apply the optimal mechanism to our reverse auction:
\begin{corollary}
    The optimal solution to (\ref{opt}) is given by Eqs.~(\ref{eq:opt_rho}) and (\ref{eq:opt_x}) with $v_0=P_{s},\,v_i = P_{si}^\text{min},\, \forall i\in N$, and $f_{i}(v_i) = f_i^{\text{min}}(P_{si}^\text{min})$, as given by Eq.~(\ref{eq:pdf_valuation}).
\end{corollary}
\begin{proof}
    This follows from Lemmas~\ref{lma:myerson} and \ref{lma:regularity}.
\end{proof}

\begin{remark}
    While the above results have been derived for lognormal fading, similar results may be derived for other statistical channel models, even if they are not regular. However, if not regular, a more complex method must be used~\cite{Myerson1981MOR}, which may be less computationally tractable.
\end{remark}

The interpretation of the auction result in the context of the underlying communication system requires brief consideration. When $\rho(v)=0$, the source ``wins" the auction (no suitable candidate is found), and if, additionally, $P_{s} \leq P_{\text{max}}$, the source communicates directly with the AP with power $x(v)=P_s$. However, if $P_{s} > P_{\text{max}}$, the source will be unable to connect to the AP, and $x(v)=P_{\text{max}}$ is no longer representative of the transmit power (which is $0$), as no communication occurs. In particular, communication will occur under the perfect info baseline but will \textit{not} occur under the Myerson auction protocol if the following three conditions are met: the source cannot communicate directly ($P_s >P_{\text{max}}$), the source wins the auction ($c_i(v_i) > P_{\text{max}}\forall i \in N$), and there exists at least one candidate, $i$, such that $P_{si}^\text{min} \leq P_{\text{max}}$. However, as we will see the experimental results of the Section~\ref{sec:numerical_examples}, the Myerson auction still greatly increases the likelihood of successful communication.

\section{Numerical Examples}\label{sec:numerical_examples}
 \vspace{-0.05in}
 To verify the validity of our proposed approach, we run several simulation experiments with system parameters as reported from experimental studies. We consider the scenario depicted in Fig.~\ref{fig:system_setup}, in which a source seeks to connect with an access point with which it has a NLoS channel. We randomly place up to 10 relay candidates with LoS channels to both the AP and the source. We then follow the auction protocol described in Section~\ref{sec:solution} to select a candidate relay (if any) and determine the transmit power of the source. Throughout, we assume all links communicate using 8-QAM modulation, with a target bit error rate (BER) of $1\mathrm{e}{-6}$, so that the target SNR in dB is given by $\gamma_{\text{th,dB}} = 33.18$ \cite{WirelessComms2005Goldsmith}. Communication noise is modeled as AWGN with power $\sigma^2 = -75\,$dBm, and the maximum transmit power of the source is $100\,$mW. The channel parameters are given as follows: $H_{\text{PL},0, dB}^{\text{L}}=0$, $H_{\text{PL},0, dB}^{\text{NL}}=-25$, $n_{\text{PL}}^{\text{L}}=2.5$, $n_{\text{PL}}^{\text{NL}}=5.76$, $\sigma_{SS,dB}^{\text{L}} = 8.66$, and $\sigma_{SS,dB}^{\text{NL}} = 9.06$ \cite[Table VI]{Hemadeh2018CST},~\cite{Azar2013ICC}. For WPT, we take $A_r=1\,\text{cm}^2$ and consider $\alpha\in\{0.1, 0.2, 0.3, 0.4\}$. For several experiments, we vary both $\sigma_{SS,dB}^{\text{L}}$ and $\sigma_{SS,dB}^{\text{NL}}$ by scaling the given values with a factor, $\Gamma\in\{0.2, 0.6, 1.0, 1.4\}$. Each reported average is taken over a set of 10,000 simulations with randomly sampled channel fading, and, unless stated otherwise, random placement of the candidates.

 We next examine system performance in terms of outage probability, source transmit power, and harvested power, while exploring system parameters including circuitry efficiency, $\alpha$, fading power, $\sigma_{SS,dB}^{\text{L}}$, and number of candidates, $n$.

 \subsection{Outage Probability}
We first study the impact of our protocol on the outage probability of the source. Fig.~\ref{fig:allResults} (a) plots the outage probability of the source for different values of $n$, $\alpha$ and fading scale factors, $\Gamma$.  In the absence of relay candidates, the source experiences an outage probability of between $1$ and $0.89$, depending on the value of $\Gamma$, but we see that the presence of as few as four relay candidates, the outage probability drops below $0.2$. This probability drops to near $0$ for higher values of $n$. Furthermore, increasing either power efficiency, $\alpha$, or fading factor, $\Gamma$, reduces outage probability. 
 \vspace{-0.05in}
\subsection{Source Transmit Power}
We next examine the transmit power used by the source as determined by the auction protocol. To gauge energy efficiency, we compare the transmit power, $P^{\text{tot}}$, of the source under the Myerson mechanism to the perfect-information transmit power, $P^{\text{lb}}$ (described in Section~\ref{sec:formal_prob_statement}), which is a lower bound on $P^{\text{tot}}$. As an upper bound, we also compare against a standard sealed-bid, second-price auction (Vickrey auction), in which the source selects the candidate relay with the lowest valuation but transmits with a power equal to the second-lowest valuation \cite{vickrey1961JoF}. Fig.~\ref{fig:allResults} (b) shows significant power savings when using the proposed buyer-optimal Myerson auction protocol compared to a standard Vickrey auction, particularly when fewer relay candidates are present. Furthermore, as the number of candidates increases, our proposed Myerson auction solution approaches the perfect-information lower bound. We also note that the performance of the Vickrey and Myerson auctions become comparable when the number of candidates is large, suggesting that in scenarios where candidates and sources have less information about each other's channels, the Vickrey auction may produce acceptable results given a sufficiently high density of UEs.
\begin{figure}
    \centering
    \includegraphics[width =\linewidth, trim={0.175in, 0.15in, 0.825in, 0.1in}, clip]{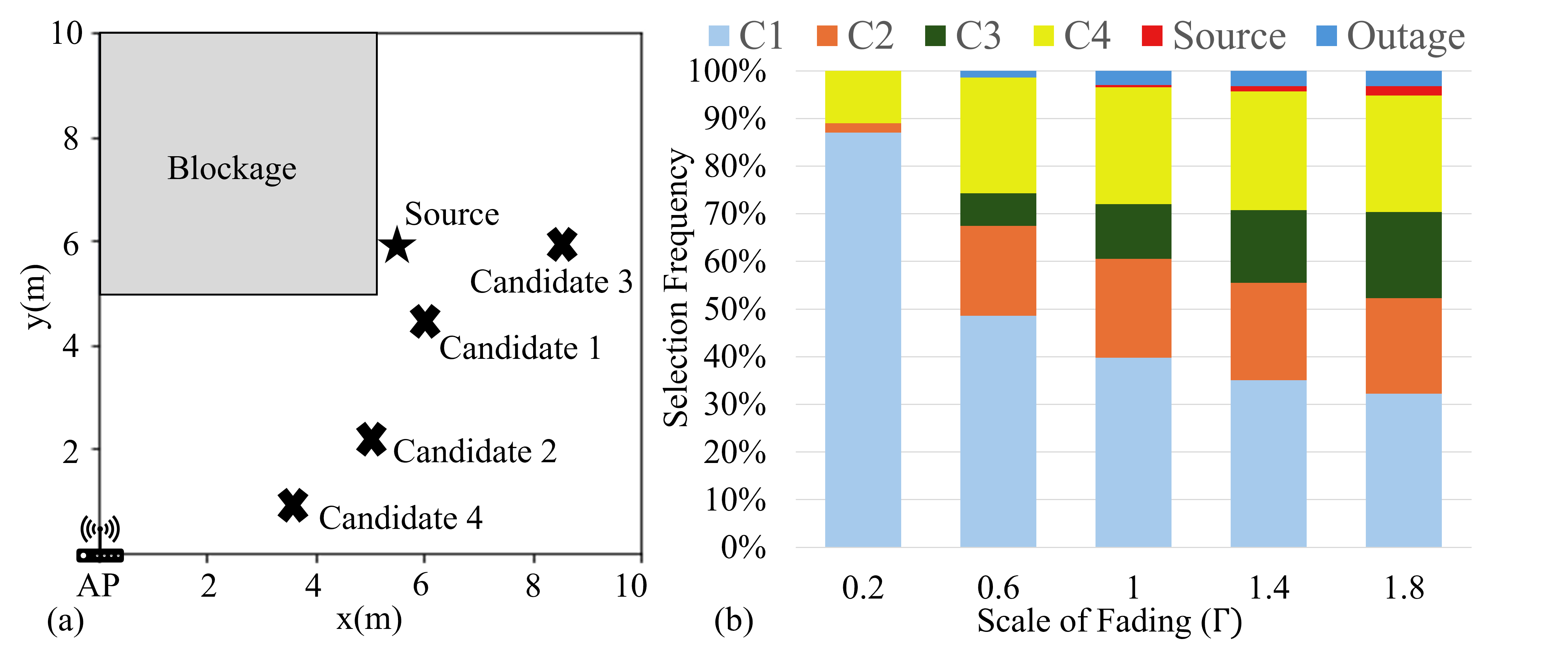}
    \vspace{-0.2in}
    \caption{(a) Sample placement of UEs for candidate selection study (b) Variation in selection likelihood across several scales ($\Gamma$) of fading. See color PDF for optimal viewing.}
    \vspace{-0.2in}
    \label{fig:scenario_selFreq}
\end{figure}

\vspace{-0.05in}
\subsection{Harvested Transmission Power}
\vspace{-0.05in}
While energy harvesting is not the chief objective of our protocol, it plays an important role in motivating the candidates to cooperate. Fig.~\ref{fig:allResults} (c) shows the average amount of power harvested by candidates for various efficiency factors, $\alpha$. As expected, more energy is harvested as the WPT efficiency is increased, and we see that the candidates receive small but significant boosts, on the order of tens of $\mu$W. However, as the number of candidates grows, competition allows the source to reduce the transmit power, $P^{\text{tot}}$, which in turn reduces the power available for harvesting.

\vspace{-0.05in}
 \subsection{Candidate Assignment}
 \vspace{-0.05in}
To better understand how the placement of a candidate affects its likelihood of winning the auction, we consider the system of $n=4$ candidates at fixed positions, as shown in Fig.~\ref{fig:scenario_selFreq} (a). For each value of $\Gamma$, we sample 10,000 realizations of channel fading and report the frequency with which each candidate is chosen, as shown in Fig.~\ref{fig:scenario_selFreq} (b). At low fading, we see that Candidate 1, whose proximity to the source results in more efficient WPT, consistently wins the auction. However, as the fading power increases, the location-dependent path loss component of the channel loses prominence, and the auction mechanism picks the other candidates more frequently.
 \section{Conclusion}
 \vspace{-0.05in}
This paper addresses the use of wireless power transfer (WPT) to incentivize energy-constrained user equipment (UE) to serve as cooperative relays in challenging non-line-of-sight communication scenarios. By proposing a cooperation-inducing protocol inspired by Myerson auction theory, we minimize the expected power used by the source while considering the self-interest of relay candidates and their private channel state information. Rigorous analysis establishes the regularity of valuations for channels modeled with lognormal fading, facilitating efficient determination of the optimal source transmit power. Extensive simulation experiments validate our framework, demonstrating a significant reduction in outage probability and power consumption, thus promoting energy-efficient cooperative relaying in wireless environments.
\vspace{-0.2in}
\bibliographystyle{IEEEtran}
\bibliography{main} 

\begin{thebibliography}{10}
\providecommand{\url}[1]{#1}
\csname url@samestyle\endcsname
\providecommand{\newblock}{\relax}
\providecommand{\bibinfo}[2]{#2}
\providecommand{\BIBentrySTDinterwordspacing}{\spaceskip=0pt\relax}
\providecommand{\BIBentryALTinterwordstretchfactor}{4}
\providecommand{\BIBentryALTinterwordspacing}{\spaceskip=\fontdimen2\font plus
\BIBentryALTinterwordstretchfactor\fontdimen3\font minus \fontdimen4\font\relax}
\providecommand{\BIBforeignlanguage}[2]{{%
\expandafter\ifx\csname l@#1\endcsname\relax
\typeout{** WARNING: IEEEtran.bst: No hyphenation pattern has been}%
\typeout{** loaded for the language `#1'. Using the pattern for}%
\typeout{** the default language instead.}%
\else
\language=\csname l@#1\endcsname
\fi
#2}}
\providecommand{\BIBdecl}{\relax}
\BIBdecl

\bibitem{IMT-2030}
ITU-R, ``{Framework and overall objectives of the future development of IMT for 2030 and beyond},'' ITU-R, Tech. Rep., 2023, {Draft New Recommendation}.

\bibitem{Hemadeh2018CST}
I.~A. Hemadeh, K.~Satyanarayana, M.~El-Hajjar, and L.~Hanzo, ``{Millimeter-Wave Communications: Physical Channel Models, Design Considerations, Antenna Constructions, and Link-Budget},'' \emph{IEEE Communications Surveys \& Tutorials}, vol.~20, no.~2, pp. 870--913, 2018.

\bibitem{Hossain2019Access}
M.~A. Hossain, R.~Md~Noor, K.-L.~A. Yau, I.~Ahmedy, and S.~S. Anjum, ``{A Survey on Simultaneous Wireless Information and Power Transfer With Cooperative Relay and Future Challenges},'' \emph{IEEE Access}, vol.~7, pp. 19\,166--19\,198, 2019.

\bibitem{Li2021IoTJ}
X.~Li, Q.~Wang, M.~Liu, J.~Li, H.~Peng, M.~J. Piran, and L.~Li, ``{Cooperative Wireless-Powered NOMA Relaying for B5G IoT Networks With Hardware Impairments and Channel Estimation Errors},'' \emph{IEEE Internet of Things Journal}, vol.~8, no.~7, pp. 5453--5467, 2021.

\bibitem{salim2022IoTaIS}
M.~M. Salim, H.~A. Elsayed, M.~S. Abdalzaher, and M.~M. Fouda, ``{RF Energy Harvesting Dependency for Power Optimized Two-Way Relaying D2D Communication},'' in \emph{2022 IEEE Int. Conf. on Internet of Things and Intelligence Systems (IoTaIS)}, 2022, pp. 297--303.

\bibitem{Xu2022ICC}
Y.~Xu, J.~Liu, H.~Takakura, Z.~Li, Y.~Ji, and N.~Shiratori, ``{Stackelberg Game-based Secure Communication in SWIPT-enabled Relaying Systems},'' in \emph{ICC 2022 - IEEE Int. Conf. Commun.}, 2022, pp. 1462--1467.

\bibitem{GTforWCN2012HanNiyatoSaadBasarHJorungnes}
Z.~Han, D.~Niyato, W.~Saad, T.~Basar, and A.~Hjorungnes, \emph{{Game Theory in Wireless and Communication Networks: Theory, Models, and Applications}}.\hskip 1em plus 0.5em minus 0.4em\relax Cambridge University Press, 2012.

\bibitem{TRO19_MuralidharanMostofi}
A.~Muralidharan and Y.~Mostofi, ``{Path Planning for Minimizing the Expected Cost Until Success},'' \emph{IEEE Transactions on Robotics}, vol.~35, no.~2, pp. 466--481, 2019.

\bibitem{Ni2019ICCT}
Y.~Ni, S.~Zhou, Q.~Wang, Y.~Zhou, and H.~Zhu, ``{Auction Game Based Phone-to-Phone Electricity Trading via Wireless Energy Transfer},'' in \emph{2019 IEEE 19th Int. Conf. Commun. Technol (ICCT)}, 2019, pp. 1213--1219.

\bibitem{Myerson1981MOR}
R.~B. Myerson, ``{Optimal Auction Design},'' \emph{Mathematics of Operations Research}, vol.~6, no.~1, pp. 58--73, 1981.

\bibitem{Samimi2016VTC}
M.~K. Samimi, G.~R. MacCartney, S.~Sun, and T.~S. Rappaport, ``{28 GHz Millimeter-Wave Ultrawideband Small-Scale Fading Models in Wireless Channels},'' in \emph{2016 IEEE 83rd Vehicular Technology Conference (VTC Spring)}, 2016, pp. 1--6.

\bibitem{Zhou2012GLOBECOM}
X.~Zhou, R.~Zhang, and C.~K. Ho, ``{Wireless information and power transfer: Architecture design and rate-energy tradeoff},'' in \emph{2012 IEEE Global Communications Conf. (GLOBECOM)}, 2012, pp. 3982--3987.

\bibitem{Liu2013TC}
L.~Liu, R.~Zhang, and K.-C. Chua, ``{Wireless Information and Power Transfer: A Dynamic Power Splitting Approach},'' \emph{IEEE Transactions on Communications}, vol.~61, no.~9, pp. 3990--4001, 2013.

\bibitem{TWC11}
M.~Malmirchegini and Y.~Mostofi, ``{On the Spatial Predictability of Communication Channels},'' \emph{IEEE Trans. on Wireless Commun.}, 2012.

\bibitem{Nisan_Roughgarden_Tardos_Vazirani_2007}
N.~Noam, T.~Roughgarden, E.~Tardos, and V.~V. Vazirani, \emph{{Algorithmic Game Theory}}.\hskip 1em plus 0.5em minus 0.4em\relax Cambridge University Press, 2007.

\bibitem{Gordon1941AMS}
R.~D. Gordon, ``{Values of Mills' Ratio of Area to Bounding Ordinate and of the Normal Probability Integral for Large Values of the Argument},'' \emph{The Ann. of Mathematical Statistics}, vol.~12, no.~3, pp. 364--366, 1941.

\bibitem{WirelessComms2005Goldsmith}
A.~Goldsmith, \emph{{Wireless Communications}}.\hskip 1em plus 0.5em minus 0.4em\relax Cambridge Univ. Press, 2005.

\bibitem{Azar2013ICC}
Y.~Azar, G.~N. Wong, K.~Wang, R.~Mayzus, J.~K. Schulz, H.~Zhao, F.~Gutierrez, D.~Hwang, and T.~S. Rappaport, ``{28 GHz propagation measurements for outdoor cellular communications using steerable beam antennas in New York city},'' in \emph{2013 IEEE Int. Conf. on Commun. (ICC)}, 2013, pp. 5143--5147.

\bibitem{vickrey1961JoF}
W.~Vickrey, ``{Counterspeculation, Auctions, and Competitive Sealed Tenders},'' \emph{The Journal of Finance}, vol.~16, no.~1, pp. 8--37, 1961.

\end{thebibliography}

\end{document}